\pdfoutput=1
\documentclass[a4paper,UKenglish,cleveref, autoref, thm-restate, numberwithinsect,thm-restate]{lipics-v2021}

\title{Internal Pattern Matching in Small Space and Applications}

\date{\empty}

\author{Gabriel Bathie}{DIENS, \'{E}cole normale sup\'{e}rieure de Paris, PSL Research University, France \and LaBRI, Université de Bordeaux, France}{gabriel.bathie@gmail.com}{0000-0003-2400-4914}{Partially supported by the grant ANR-20-CE48-0001 from the French National Research Agency (ANR) and a Royal Society International Exchanges Award.}

\author{Panagiotis Charalampopoulos}{Birkbeck, University of London, UK}{pcharalampo@gmail.com}{0000-0002-6024-1557}{Research visits during which some of the presented ideas were conceived were funded by a Royal Society International Exchanges Award.}

\author{Tatiana Starikovskaya}{DIENS, \'{E}cole normale sup\'{e}rieure de Paris, PSL Research University, France}{tat.starikovskaya@gmail.com}{0000-0002-7193-9432}{Partially supported by the grant ANR-20-CE48-0001 from the French National Research Agency (ANR) and a Royal Society International Exchanges Award.}

\authorrunning{G.~Bathie, P.~Charalampopoulos, and T.~Starikovskaya}

\Copyright{Gabriel Bathie, Panagiotis Charalampopoulos, and Tatiana Starikovskaya}

\ccsdesc[500]{Theory of computation~Pattern matching}

\keywords{internal pattern matching, longest common substring, small-space algorithms}

\acknowledgements{We would like to dedicate this work to our dear friend and colleague Pawe{\l} Gawrychowski on the occasion of his 40th birthday. We thank Solon Pissis for helpful suggestions.}%optional

\nolinenumbers %uncomment to disable line numbering

\EventEditors{Shunsuke Inenaga and Simon J. Puglisi}
\EventNoEds{2}
\EventLongTitle{35th Annual Symposium on Combinatorial Pattern Matching (CPM 2024)}
\EventShortTitle{CPM 2024}
\EventAcronym{CPM}
\EventYear{2024}
\EventDate{June 25--27, 2024}
\EventLocation{Fukuoka, Japan}
\EventLogo{}
\SeriesVolume{296}
\ArticleNo{26}

\hideLIPIcs

\usepackage{xspace}
\usepackage[T1]{fontenc}
\usepackage[utf8]{inputenc}
\usepackage{xcolor}
\usepackage{amsmath}
\usepackage{amssymb}
\usepackage{amsthm}
\usepackage{thmtools}
\usepackage{ifthen}
\usepackage{tikz}
\usetikzlibrary{calc} % + notation on coordinates
\usetikzlibrary{decorations.pathreplacing} % curly brakets
\usetikzlibrary{patterns}
\usetikzlibrary{math} %for tikzmath and variables
\usepackage{comment}
\usepackage{xcolor,stackengine,graphicx,scalerel,xspace}
\usepackage{multirow}

\newcommand{\cO}{O}
\newcommand{\tO}{\tilde{O}}
\newcommand{\cOtilde}{\tilde{O}}

\newcommand{\LCE}{\textsf{LCE}\xspace}
\newcommand{\LCS}{\textsf{LCS}\xspace}
\newcommand{\IPM}{\textsf{IPM}\xspace}
\newcommand{\LCSu}{\textsf{LCSuf}\xspace}
\newcommand{\CPM}{\textsf{CPM}\xspace}
\newcommand{\per}{\textsf{per}}
\newcommand{\rev}{\mathrm{rev}}

\def\dd{\mathinner{.\,.}}
\renewcommand{\P}{\mathcal{P}}
\newcommand{\A}{\mathcal{A}}
\newcommand{\X}{\mathcal{X}}
\newcommand{\fd}{\textsf{SST}}
\newcommand{\bd}{\textsf{RSST}}
\newcommand{\AI}{\textsc{Augmented Index}\xspace}

\def\pillar{{\tt PILLAR}\xspace}

\newcommand{\defproblem}[3]{
\vspace{2mm}
\noindent\fbox{
  \begin{minipage}{0.94\textwidth}
    \textsc{\large #1}\\
    {\bf{Input:}} #2  \\
    {\bf{Output:}} #3
  \end{minipage}
  }
\vspace{2mm}
}

\newcommand{\defCCproblem}[4]{
\vspace{2mm}
\noindent\fbox{
  \begin{minipage}{0.94\textwidth}
    \textsc{\large #1}\\
    {\bf{Alice}} holds #2  \\
    {\bf{Bob}} holds #3  \\
    {\bf{Output:}} #4
  \end{minipage}
  }
\vspace{2mm}
}

\newcommand{\defDSproblem}[3]{
\vspace{2mm}
\noindent\fbox{
   \begin{minipage}{0.94\textwidth}
   \textsc{#1}\\
   {\bf{Input:}} #2  \\
   {\bf{Query:}} #3
   \end{minipage}
   }
   \vspace{2mm}
}

\newtheorem{fact}[theorem]{Fact}

\begin{document}

\maketitle
\begin{abstract}
  In this work, we consider pattern matching variants in small space, that is, in the read-only setting, where we want to bound the space usage on top of storing the strings.
  Our main contribution is a space-time trade-off for the \textsc{Internal Pattern Matching} (IPM) problem, where the goal is to construct a data structure over a string $S$ of length $n$ that allows one to answer the following type of queries: Compute the occurrences of a fragment $P$ of $S$ inside another fragment $T$ of $S$, provided that $|T| < 2|P|$.
  For any $\tau \in [1 \dd n/\log^2 n]$, we present a nearly-optimal $\cOtilde(n/\tau)$-size\footnote{Throughout this work, the $\tilde{O}(\cdot)$ notation suppresses factors polylogarithmic in the input-size.} data structure that can be built in $\cOtilde(n)$ time using $\cOtilde(n/\tau)$ extra space, and answers IPM queries in $\cO(\tau+\log n \log^3  \log n)$ time.
  IPM queries have been identified as a crucial primitive operation for the analysis of algorithms on strings.
  In particular, the complexities of several recent algorithms for approximate pattern matching are expressed with regards to the number of calls to a small set of primitive operations that include IPM queries; our data structure allows us to port these results to the small-space setting.
  We further showcase the applicability of our IPM data structure by using it to obtain space-time trade-offs for the longest common substring and circular pattern matching problems in the \emph{asymmetric streaming} setting.
\end{abstract}

\clearpage
\setcounter{page}{1}

\section{Introduction}
  
In the fundamental text indexing problem, the task is to preprocess a text $T$ into a data structure (index) that can answer the following queries efficiently: Given a pattern $P$, find the occurrences of $P$ in $T$. The \textsc{Internal Pattern Matching} problem (\IPM) is a variant of the text indexing problem, where both the pattern $P$ and the text $T$ are fragments of a longer string~$S$, given in advance. 

Introduced in 2009~\cite{DBLP:journals/tcs/KellerKFL14}, \IPM queries are a cornerstone of the family of internal queries on strings. The list of internal queries, primarily executed through \IPM queries, comprises of period queries, prefix-suffix queries, periodic extension queries, and cyclic equivalence queries; see~\cite{DBLP:journals/corr/KociumakaRRW13,doi:10.1137/1.9781611973730.36,koc18}.
Other problems that have been studied in the internal setting include shortest unique substring~\cite{a13110276}, longest common substring~\cite{DBLP:journals/algorithmica/AmirCPR20}, suffix rank and selection~\cite{doi:10.1137/1.9781611973730.39,koc18}, BWT substring compression~\cite{doi:10.1137/1.9781611973730.39}, shortest absent string~\cite{BADKOBEH2022271}, dictionary matching~\cite{DBLP:journals/corr/abs-2312-11873,DBLP:journals/algorithmica/Charalampopoulos21,DBLP:conf/cpm/Charalampopoulos20}, string covers~\cite{DBLP:conf/spire/CrochemoreIRRSW20}, masked prefix sums~\cite{10.1007/978-3-031-20643-6_16}, circular pattern matching~\cite{DBLP:conf/cpm/IliopoulosKRRWZ23}, and longest palindrome~\cite{10.1007/978-3-031-27051-2_12}.

The primary distinction between the classical and internal string queries lies in how the pattern is handled during queries. In classical queries, the input is explicitly provided at query time, whereas in internal queries, the input is specified in constant space via the endpoints of fragments of string $S$. This distinction enables notably faster query times in the latter setting, as there is no need to read the input when processing the query. This characteristic of \IPM and similar internal string queries renders them particularly valuable for bulk processing of textual data. This is especially advantageous when $S$ serves as input for another algorithm, as illustrated by multiple direct and indirect (via other internal queries) applications of \IPM: pattern matching with variables~\cite{10.1007/978-3-319-67428-5_22,10.1145/3369935}, detection of gapped repeats and subrepetitions~\cite{KOLPAKOV20171,DBLP:journals/mst/GawrychowskiIIK18}, approximate period recovery~\cite{AMIR20182,DBLP:journals/algorithmica/AmirBKLS22}, computing the longest unbordered substring~\cite{kociumaka_et_al:LIPIcs.ISAAC.2018.70}, dynamic repetition detection~\cite{amir_et_al:LIPIcs.ESA.2019.5}, computing string covers~\cite{DBLP:conf/spire/CrochemoreIRRSW20}, identifying two-dimensional maximal repetitions, enumeration of distinct substrings~\cite{DBLP:conf/spire/Charalampopoulos20}, dynamic longest common substring~\cite{DBLP:journals/algorithmica/AmirCPR20}, approximate pattern matching~\cite{9317938,9996673}, approximate circular pattern matching~\cite{DBLP:journals/jcss/Charalampopoulos21,DBLP:conf/esa/Charalampopoulos22}, (approximate) pattern matching with wildcards~\cite{DBLP:journals/corr/abs-2402-07732}, RNA folding~\cite{das_et_al:LIPIcs.ICALP.2022.49}, and the language edit distance problem for palindromes and squares~\cite{DBLP:conf/isaac/BathieKS23}. 

Below we assume $|T| < 2 |P|$, which guarantees that the set of occurrences of $P$ in $T$ forms an arithmetic progression and can be thus represented in $\cO(1)$ space. 

With no preprocessing ($\cO(1)$ extra space), \IPM queries on a string $S$ of length $n$ can be answered in $\cO(n)$ time by a constant-space pattern matching algorithm (see~\cite{DBLP:journals/tcs/BreslauerGM13} and references therein). On the other side of the spectrum, Kociumaka, Radoszewski, Rytter, and Wale\'{n}~\cite{DBLP:journals/corr/KociumakaRRW13} showed that for every string $S \in [0 \dd \sigma]^n$, there exists a data structure of size $\cO(n/\log_\sigma n)$ which answers 
\IPM queries in optimal $\cO(1)$ time and can be constructed in $\cO(n/\log_\sigma n)$ time given the packed representation of $S$ (meaning that $S$ divided into blocks of $\log_\sigma n$ consecutive letters, and every block is stored in one machine word). The problem has been equally studied in the compressed and dynamic settings~\cite{9317938,DBLP:conf/stoc/KempaK22,DBLP:conf/focs/KempaK20}.

\subsection{Our Main Contribution: Small-space \IPM}
As our main contribution, we provide a trade-off between the constant-space and $\cO(n)$ query time and the $\cO(n/\log_\sigma n)$-space and constant query time data structures.
We consider the \IPM problem in the read-only setting, where one assumes random read-only access to the input string(s) and only accounts for the extra space, that is, the space used by the algorithm/data structure on top of the space needed to store the input. 

\begin{restatable}{corollary}{ipm}\label{cor:ipm}
Suppose that we have read-only random access to a $n$-length string $S$ of length $n$ over an integer alphabet. 
For any integer $\tau = \cO(n / \log^2 n)$, there is a data structure that can be built using $\cO(n\log_{n/\tau} n + (n/\tau) \cdot \log^4 n \log \log n)$ time using $\cO((n/\tau) \cdot \log n (\log \log n)^3)$ extra space and can answer the following internal pattern matching queries in time $\cO(\tau+\log n \log^3  \log n)$: 
given $p,p',t,t' \in [1\dd n]$ such that $t'-t \leq 2(p'-p)$, return all occurrences of $P = S[p \dd p']$ in $T = S[t \dd t']$.
\end{restatable}

Our data structure is nearly optimal: First, when $n/\tau$ is polynomial, the construction time is linear; and secondly, the product of the query time and space of our data structure is optimal up to polylogarithmic factors (\textbf{\cref{lemma:ipm-lb}}). 

\noindent \textbf{Technical overview for \IPM queries.} Our solution relies heavily on utilizing the concept of $\tau$-partitioning sets, as introduced by Kosolobov and Sivukhin~\cite{DBLP:journals/corr/abs-2105-03782}. For a string of length~$n$, a $\tau$-partitioning is a subset of $\cO(n/\tau)$ positions that satisfies some density and consistency criteria. We use the positions of such a set as anchor points for identifying pattern occurrences, provided that the pattern avoids a specific periodic structure. To detect these anchored occurrences, we employ sparse suffix trees alongside a three-dimensional range searching structure. In cases where the pattern does not avoid said periodic structure, we employ a different strategy, leveraging the periodic structure to construct the necessary anchor points.

We next provide a brief comparison of the outlined approach with previous work. String anchoring techniques have been proven very useful in and been developed for text indexing problems, such as the longest common extension (\LCE) problem, in small space~\cite{DBLP:journals/corr/abs-2105-03782,DBLP:conf/soda/BirenzwigeGP20}.
One of the most technically similar works to ours is that of Ben-Nun et al.~\cite{DBLP:conf/cpm/Nun0KK20} who considered the problem of computing a long common substring of two input strings in small space.
They use an earlier variant of $\tau$-partitioning sets, due to Birenzwige et al.~\cite{DBLP:conf/soda/BirenzwigeGP20}, that has slightly worse guarantees than that of Kosolobov and Shivukhin \cite{DBLP:journals/corr/abs-2105-03782}.
The construction of anchors for substrings with periodic structure is quite similar to that of Ben-Nun et al.~\cite{DBLP:conf/cpm/Nun0KK20}.
After computing a set of anchors, they aim to identify a synchronised pair of anchors that yields a long common substring; they achieve this via mergeable AVL trees.
As \IPM queries need to be answered in an online manner, we instead construct an appropriate orthogonal range searching data structure over a set of points that correspond to anchors. Using orthogonal range searching is a by-now classical approach for text indexing, see~\cite{DBLP:conf/birthday/Lewenstein13} for a survey.

\subsection{Applications}

Several internal queries reduce to \IPM queries, and hence we obtain efficient implementations of them in the small-space setting.
Additionally, we port several efficient approximate pattern matching algorithms to the small-space setting since \IPM was the only primitive operation that they rely on that did not have an efficient small-space implementation to this day.
See \cref{sec:apps} for details on these applications.

\noindent \textbf{\textsc{Longest Common Substring} (\LCS). }
The \LCS problem is formally defined as follows.

\defproblem{Longest Common Substring (\LCS)}{Strings $S$ and $T$ of length at most $n$.}{The length of a longest string that appears as a (contiguous) fragment in both $S$ and $T$.}

The length of the longest common substring is one of the most popular string-similarity measures. The by-now classical approach to the \LCS problem is to construct the suffix tree of~$S$ and $T$ in $\cO(n)$ time and space. The longest common substring of the two strings appears as a common prefix of a pair of suffixes of $S$ and $T$ and hence its length is the maximal string-depth of a node of the suffix tree with leaf-descendants corresponding to suffixes of both strings; this node can be found in $\cO(n)$ time in a bottom-up manner. 

Starikovskaya and Vildh\o{}j~\cite{DBLP:conf/cpm/StarikovskayaV13} were the first to consider the problem in the read-only setting.
They 
showed that for any $n^{2/3} < \tau \le n$, the problem can be solved in $
\cO(\tau)$ extra space and $\cO(n^2/\tau)$ time. Kociumaka et al.~\cite{kociumaka2014sublinear} extended their bound to all $1\le\tau\le n$, which in particular resulted in a constant-space read-only algorithm running in time $\cOtilde(n^2)$. 

In an attempt to develop even more space-efficient algorithms for the \LCS problem, it might be tempting to consider the streaming setting, which is particularly restrictive: in this setting, one assumes that the input arrives letter-by-letter, as a stream, and must account for all the space used. Unfortunately, this setting does not allow for better space complexity: any streaming algorithm for \LCS, even randomised, requires $\Omega(n)$ bits of space (\textbf{\cref{thm:lb-lcs}}). In the asymmetric streaming setting, which is slightly less restrictive and was introduced by Andoni et al.~\cite{5671209} and Saks and Seshadhri~\cite{doi:10.1137/1.9781611973105.122}, the algorithm has random access to one string and sequential access to the other. Mai et al.~\cite{mai2021optimal} showed that in this setting, \LCS can be solved in $\tO(n^2)$ time and $\cO(1)$ space. By utilising 
(a slightly more general variant of) \IPM queries, we extend their result and show that for every $\tau \in [\sqrt n \log n (\log \log n)^3 \dd n]$, there is an asymmetric streaming algorithm that solves the \LCS problem in $\cO(\tau)$ space and $\tO(n^2/\tau)$ time (\textbf{\cref{thm:lcs-aux}}). Note that these bounds almost match the bounds of Kociumaka et al.~\cite{kociumaka2014sublinear}, while the setting is stronger. 

\noindent \textbf{\textsc{Circular Pattern Matching} (\CPM).}
The \CPM problem is formally defined as follows.

\defproblem{Circular Pattern Matching (\CPM)}{A pattern $P$ of length $m$, a text $T$ of length $n$.}{All occurrences of rotations of $P$ in $T$.}

The interest in occurrences of rotations of a given pattern is motivated by applications in Bioinformatics and Image Processing: in Bioinformatics, the starting position of a biological sequence can vary significantly due to the arbitrary nature of sequencing in circular molecular structures or inconsistencies arising from different standards of linearization applied to sequence databases; and in Image Processing, the contour of a shape can be represented using a directional chain code, which can be viewed as a circular sequence, particularly when the orientation of the image is irrelevant~\cite{AYAD201781}. 

For strings over an alphabet of size $\sigma$, the classical read-only solution for \CPM via the suffix automaton of $P \cdot P$ runs in $\cO(n \log \sigma)$ time and uses $\cO(m)$ extra space~\cite{Lothaire_2005}. Recently, Charalampopoulos et al. showed a simple $\cO(n)$ time and $\cO(m)$ extra space solution. The problem has been also studied from the practical point of view~\cite{10.1007/978-3-319-02309-0_59,FREDRIKSSON2009579,10.1093/comjnl/bxt023} and in the text indexing setting~\cite{Iliopoulos2017,10.1007/978-3-642-25591-5_69,10.1007/978-3-642-38905-4_15}. 

It is not hard to see that the \CPM and the \LCS problems are closely related: occurrences of rotations of $P$ in $T$ are exactly the common substrings of $P \cdot P$ and $T$ of length $m$.
Implicitly using this observation, we show that
in the streaming setting, the \CPM problem requires~$\Omega(m)$ bits of space (\textbf{\cref{th:streaming-lb-cpm}}) and that
in the asymmetric streaming setting, for every $\tau \in [\sqrt m \log m (\log \log m)^3 \dd m]$, there exists an algorithm that solves the \CPM problem in time $\cOtilde(mn/\tau)$ using $\cO(\tau)$ extra space (\textbf{\cref{cor:cpm}}). 
Finally, in the read-only setting, we give an \emph{online} $\cO(n)$-time, $\cO(1)$-space algorithm (\textbf{\cref{thm:readonly}}).

\section{Preliminaries}
For integers $i,j\in \mathbb{Z}$, denote $[i\dd j] =
\{k \in \mathbb{Z} : i \le k \le j\}$, $[i
\dd j)=\{k \in \mathbb{Z} : i \le k <
j\}$.
We consider an alphabet $\Sigma = \{1,2,\ldots,\sigma\}$ of size polynomially bounded in the length of the input string(s).
The elements of the alphabet are called letters, and a string is a finite sequence of letters. 
For a string $T$ and an index $i\in [1\dd n]$, the $i$-th letter of $T$ is denoted by $T[i]$.
We use $|T| = n$ to denote the length of $T$. For two strings $S,T$, we use $ST$ or $S \circ T$ indifferently to denote their concatenation $S[1] \cdots S[|S|] T[1] \cdots T[|T|]$.
For integers $i,j$, 
$T[i\dd j]$ denotes the \emph{fragment} $T[i] T[{i+1}]\cdots T[j]$ of~$T$ if $1 \le i \le j\le n$
and the empty string~$\varepsilon$ otherwise. We extend this notation in a natural way to $T[i \dd j+1) = T[i\dd j] = T(i-1 \dd j]$.   
When $i=1$ or $j=n$, we omit these indices, i.e., $T[\dd j] = T[1\dd j]$
and $T[i\dd ] = T[i\dd n]$.
A string $P$ is a \emph{prefix} of $T$ if there exists $j\in [1\dd n]$
such that $P = T[\dd j]$, and a \emph{suffix} of~$T$ if there exists $i\in [1\dd n]$
such that $P = T[i\dd ]$. 
We denote the reverse of a string~$T$ by $\rev(T) = T[n]T[n-1] \cdots T[2]T[1]$.
For an integer $\Delta \in [1 \dd n]$, we say that a string $T[\Delta+1 \dd n] \circ T[1 \dd \Delta]$ is a \emph{rotation} of $T$. 
A fragment $T[i\dd j]$ of a string $T$ is called an \emph{occurrence} of a string $P$ if $T[i\dd j] = P$; in this case, we say that $P$ \emph{occurs} at position $i$ of~$T$. A positive integer $\rho$ is a \emph{period} of a string $T$ if $T[i] = T[i+\rho]$ for all $i \in [1 \dd |T|-\rho]$.
The smallest period of $T$ is referred to as \emph{the period} of $T$ and is denoted by $\per(T)$. If $\per(T) \leq |T|/2$, $T$ is called \emph{periodic}. 

\begin{fact}[Corollary of the Fine--Wilf periodicity lemma~\cite{fine1965uniqueness}]\label{cor:progression}
The starting positions of the occurrences of a pattern $P$ in a text $T$ form $\cO(|T|/|P|)$ arithmetic progressions with difference $\per(P)$.
\end{fact}

We assume a reader to be familiar with basic data structures for string processing, see, e.g., \cite{Lothaire_2005}. Recall that a suffix tree for a string $S$ is essentially a compact trie representing the set of all suffixes of $S$, whereas a sparse suffix tree contains only a subset of these suffixes.

\begin{fact}[{\cite[Theorem 3]{DBLP:journals/corr/abs-2105-03782}}]
Suppose that we have read-only random access to a string $S$ of length $n$ over an integer alphabet.
For any integer $b = \Omega(\log^2 n)$, one can construct in $\cO(n\log_b n)$ time and $\cO(b)$ space the sparse suffix tree for arbitrarily chosen $b$ suffixes.\label{thm:sst}
\end{fact}

\begin{fact}[{\cite{DBLP:journals/tcs/BreslauerGM13}}]\label{fact:pm-ro}
There is a read-only online algorithm that finds all occurrences of a pattern~$P$ of length $m$ in a text $T$ of length $n \ge m$ in $\cO(n)$ time and $\cO(1)$ space.
\end{fact}

\begin{fact}[{\cite[Lemma 6]{DBLP:conf/esa/0001GGK15}}]\label{lem:run_per}
Given read-only random access to a string $S$ of length $n$, one can decide in $\cO(n)$ time and $\cO(1)$ space if $S$ is periodic and, if so, compute $\per(S)$.
\end{fact}

\begin{fact}[{\cite{DBLP:journals/jal/Duval83}}]\label{lem:duval}
Given read-only random access to a string $S$ of length $n$, the lexicographically smallest rotation of a string $S$ can be computed in $\cO(n)$ time and $\cO(1)$ space.
\end{fact}

\subparagraph*{Static predecessor.} For a static set, a combination of x-fast tries \cite{DBLP:journals/ipl/Willard83} and deterministic dictionaries \cite{DBLP:conf/icalp/Ruzic08} yields the following efficient deterministic data structure; cf.~\cite{DBLP:conf/cpm/0001G15}.

\begin{fact}[{\cite[Proposition 2]{DBLP:conf/cpm/0001G15}}]\label{lem:pred}
A sorted static set $Y \subseteq [1\dd U]$ can be preprocessed in $\cO(|Y|)$ time and space so that predecessor queries can be performed in $\cO(\log \log |U|)$ time.
\end{fact}

\subparagraph*{Weighted ancestor queries.}
Let $\mathcal{T}$ be a rooted tree with integer weights on nodes. A \emph{weighted ancestor} query for a node $u$ and weight $d$ must return the highest ancestor of $u$ with weight at least $d$.

\begin{fact}[\cite{DBLP:journals/talg/AmirLLS07}]\label{lem:WAQ}
Let $\mathcal{T}$ be a rooted tree of size $n$ with integer weights on nodes. Assume that each weight is at most $n$, with the weight of the root being zero, and the weight of every non-root node being strictly larger than its parent's weight. $\mathcal{T}$ can be preprocessed in $\cO(n)$ time and space so that weighted ancestor queries on it can be performed in $\cO(\log \log n)$ time.
\end{fact}

If $\mathcal{T}$ is the suffix tree of a string and the weights are the string-depths of the nodes, this result can be improved further:

\begin{fact}[\cite{belazzougui_et_al:LIPIcs.CPM.2021.8}]\label{lem:WAQ-st}
The suffix tree $\mathcal{T}$ of a string of length $n$ can be preprocessed in $\cO(n)$ time and $\cO(n)$ space so that weighted ancestor queries on it can be performed in $\cO(1)$ time.
\end{fact}

\subparagraph*{3D range emptiness.}
A three-dimensional orthogonal range emptiness query asks whether a range $[a_1 \times a_2] \times [b_1 \times b_2] \times [c_1 \times c_2]$ is empty.

\begin{fact}[{\cite[Theorem 2]{DBLP:conf/cocoon/KarpinskiN09}}]\label{lem:3d}
There exists a data structure that answers three-dimensional orthogonal range emptiness queries on a set of $n$ points from a $[U] \times [U] \times [U]$ grid in $\cO(\log \log U + (\log \log n)^3)$ time, uses $\cO(n \log n(\log \log n)^3)$ space, and can be constructed in $\cO(n \log^4 n \log \log n)$ time.
If the query range is not empty, the data structure also outputs a point from it. 
\end{fact}

\begin{remark}
Better space vs.~query-time tradeoffs than the above are known for the 3D range emptiness problem; cf~\cite{DBLP:conf/compgeom/ChanLP11} and references therein.
We opted for the data structure encapsulated of \cref{lem:3d} due to its efficient construction algorithm.
Note that a data structure capable of reporting all points in an orthogonal range over a $[U] \times [U] \times [U]$ grid with $n$ points in time $\cO(Q_1(U,n) + Q_2(U,n) \cdot  |\textsf{output}|)$ can answer range emptiness queries, also returning a witness in the case the range is not empty, in time $\cO(Q_1(U,n) + Q_2(U,n))$.
\end{remark}

\section{Internal Pattern Matching}\label{sec:ipm}

We consider a slightly more powerful variant of \IPM queries, as required by our applications. A reader that is only interested in \IPM queries can focus on the case when $a = \varepsilon$.

\defDSproblem{Extended \IPM (Decision)}{A string $S$ of length $n$ over an integer alphabet to which we have read-only random access.}{Given $p,p',t,t' \in [1 \dd n]$ and $a \in \Sigma \cup \{\varepsilon \}$, return whether $P:=S[p \dd p']a$ occurs in $T:=S[t\dd t']$ and, if so, return a witness occurrence.}

Our solution for \textsc{Extended \IPM (Decision)} heavily relies on a solution for the following auxiliary problem.

\defDSproblem{Anchored \IPM}{A string $S$ of length $n$ over an integer alphabet $\Sigma$ to which we have read-only random access and a set $\A \subseteq [1 \dd n]$.}{Given $p,x,p',t,t' \in [1 \dd n]$ with $p\leq x\leq p'$, $x\in \A$, and $a \in \Sigma \cup \{\varepsilon \}$,
for $P:=S[p \dd p']a$,
decide whether there exists an occurrence of $P$ at some position $j \in [t \dd t'-|P|+1]$ such that $j + (x-p) \in \A$ and, if so, return a witness.}

\begin{lemma}\label{lem:sst3d}
There exists a data structure for the \textsc{Anchored \IPM} problem that can be built using $\cO(n\log_{|\A|} n)+\cO(|\A| \log^4 |\A| \log \log |\A|)$ time and $\cO(|\A| \log |\A| (\log \log |\A|)^3)$ extra space, and answers queries in $\cO(\log^3 \log n)$ time.
\end{lemma}
\begin{proof}
For an integer $y \in [1 \dd n]$, denote $P_y:=\rev(S[ \dd y))$ and $S_y:=S[y\dd ]$.
Consider a family $\X :=\{ (P_y\$, S_y\$) : y \in \A \}$ of pairs of strings, where $\$ \not\in \Sigma$ is a letter lexicographically smaller than all others.
Using \cref{thm:sst}, we build a sparse suffix tree $\bd$ for the first components of the elements of $\X$ and a sparse suffix tree $\fd$ for the second components of the elements of $\X$.

Consider a three-dimensional grid $[1 \dd n] \times [1 \dd n] \times [1 \dd n]$. In this grid, create a set~$\Pi$ of points, which contains, for each element $(P_y\$, S_y\$)$ of $\X$, a point $(\textsf{rank}_{\rev}(y), \textsf{rank}(y), y)$, where $\textsf{rank}_{\rev}(y)$ is the lexicographic rank of $P_y\$$ among the first components of the elements of $\X$ and $\textsf{rank}(y)$ is the lexicographic rank of $S_y\$$ among the second components of the elements of $\X$.

Upon a query, we first retrieve the leaves corresponding to $P_x\$$ and $S_x\$$ in $\bd$ and $\fd$, respectively.
This can be done in $\cO(\log\log n)$ time with the aid of~\cref{lem:pred} built over the elements of $\A$, with $x\in \A$ storing pointers to the corresponding leaves as satellite information.
Next, we retrieve the (possibly implicit) nodes $u$ and $v$ corresponding to $\rev(S[p \dd x))$ in $\bd$ and $S[x \dd p']a$ in $\fd$, respectively.
This can be done in $\cO(\log\log n)$ time after an $\cO(|\A|)$-time preprocessing of (a) the two trees according to \cref{lem:WAQ} and (b) the edge-labels of the outgoing edges of each node using~\cref{lem:pred}.
Now, it suffices to check if there is some integer~$j$ such that the leaf corresponding to $P_j\$$ is a descendant of~$u$, the leaf corresponding to $S_j\$$ is a descendant of~$v$, and $j \in [t + (x-p) \dd t' - (p'+|a|-x)]$.
After a linear-time bottom-up preprocessing of $\bd$ and $\fd$, we can retrieve in $\cO(1)$ time the following ranges:
\begin{itemize}
\item $R_1 = \{\textsf{rank}_{\rev}(y) : \text{the node of } \bd \text{ corresponding to } P_y\$ \text{ is a descendant of } u \}$;
\item $R_2 = \{\textsf{rank}(y) : \text{the node of } \fd \text{ corresponding to } S_y\$ \text{ is a descendant of } v \}$.
\end{itemize}
The query then reduces to deciding whether the orthogonal range $R_1 \times R_2 \times [t + (x-p) \dd t' - (p'+|a|-x)]$ contains any point in $\Pi$, and returning a witness if it does.
We can do this efficiently by building the data structure encapsulated in \cref{lem:3d} for $\Pi$:
the query time is $\cO(\log^3 \log n)$, while the construction time is $\cO(n\log_{|\A|} n)+\cO(|\A| \log^4 |\A| \log \log |\A|)$ and the space usage is $\cO(|\A| \log |\A| (\log \log |\A|)^3)$.
\end{proof}

For an integer parameter $\tau$, we next present a data structure for \textsc{Extended \IPM (Decision)} that uses $\cOtilde(n/\tau)$ space on top of the space required to store $S$ and answers queries in nearly-constant time provided that $P$ is of length greater than~$5\tau$.
We achieve this result using the so-called $\tau$-partitioning sets of Kosolobov and Sivukhin \cite{DBLP:journals/corr/abs-2105-03782} as \emph{anchors} for the occurrences if $P$ avoids a certain periodic structure, and by exploiting said periodic structure to construct anchors in the remaining case.

\begin{definition}[$\tau$-partitioning set]\label{def:partition-set}
Given an integer $\tau \in [4\dd n/2]$, a set of positions $\P \subseteq [1\dd n]$ is called a \emph{$\tau$-partitioning set} if it satisfies the following properties:
\begin{enumerate}[(a)]
\item if $S[i{-}\tau \dd i{+}\tau] = S[j{-}\tau \dd j{+}\tau]$ for $i,j \in [\tau+1 \dd n{-}\tau]$, then $i \in \P$ if and only if $j \in \P$; \label{cond:a}
\item if $S[i \dd i{+}\ell] = S[j \dd j{+}\ell]$, for $i,j \in \P$ and some $\ell \ge 0$, then, for each $d \in [0 \dd \ell{-}\tau)$, $i + d \in \P$ if and only if $j + d \in \P$; \label{cond:b}
\item if $i,j \in [1 \dd n]$ with $j - i > \tau$ and $(i\dd j) \cap \P = \emptyset$, then $S[i \dd j]$ has period at most~$\tau / 4$. \label{cond:c}
\end{enumerate}
\end{definition}

\begin{theorem}[{\cite{DBLP:journals/corr/abs-2105-03782}}]
Suppose that we have read-only random access to a string $S$ of length~$n$ over an integer alphabet.
For any integer $\tau \in [4\dd \cO(n / \log^2 n)]$ and $b = n / \tau$, one can construct in $\cO(n\log_b n)$ time and $\cO(b)$ extra space a $\tau$-partitioning set $\P$ of size~$\cO(b)$. The set
$\P$ additionally satisfies the property that if a fragment $S[i\dd j]$ has period at most~$\tau / 4$, then $\P \cap [i + \tau \dd j - \tau] = \emptyset$.
\label{thm:part}
\end{theorem}

\begin{definition}[$\tau$-runs]\label{def:tau-runs}
A fragment $F$ of a string $S$ is a $\tau$-run if and only if $|F| > 3\tau$, $\per(F) \leq \tau/4$, and $F$ cannot be extended in either direction without its period changing.
The Lyndon root of a $\tau$-run $R$ is the lexicographically smallest rotation of $R[1 \dd \per(R)]$.
\end{definition}

The following fact follows from the proof of Lemma 10 in the full version of \cite{DBLP:conf/esa/Charalampopoulos21}, where the definition of $\tau$-runs is slightly different, but captures all of our $\tau$-runs.

\begin{fact}[cf.~{\cite[proof of Lemma 10]{DBLP:conf/esa/Charalampopoulos21}}]\label{fact:runs_overlap}
Two $\tau$-runs can overlap by at most $\tau/2$ positions. The number of $\tau$-runs in a string of length $n$ is $\cO(n/\tau)$.
\end{fact}

\begin{lemma}\label{lemma:runs_compute}
Suppose that we have read-only random access to a string $S$ of length $n$ over an integer alphabet.
For any integer $\tau \in [4\dd \cO(n / \log^2 n)]$, all $\tau$-runs in $S$ can be computed and grouped by Lyndon root in $\cO(n\log_b n)$ time using $\cO(b)$ extra space, where $b=n/\tau$. 
Within the same complexities, we can compute, for each $\tau$-run, the first occurrence of its Lyndon root in it.
\end{lemma}
\begin{proof}
We first compute a $\tau$-partitioning set $\P$ for $S$ using \cref{thm:part}.
Due to Property~\ref{cond:c}, its converse that is stated in \cref{thm:part}, and \cref{fact:runs_overlap} there is a natural injection from the $\tau$-runs to the maximal fragments of length at least $\tau$ that do not contain any position in~$\P$~--- the $\tau$-run corresponding to such a maximal fragment may extend for $\tau$ more positions in each direction.
We can find the period of each maximal fragment in time proportional to its length using $\cO(1)$ extra space due to \cref{lem:run_per}.
We then try to extend the maximal fragment to a $\tau$-run using $\cO(\tau)$ letter comparisons. 
Additionally, we compute the Lyndon root of each computed $\tau$-run $R$ in $\cO(\tau)=\cO(|R|)$ time by applying \cref{lem:duval} to $R[1 \dd \per(R)]$.
The first occurrence of the Lyndon root in the $\tau$-run can be computed in constant time since we know which rotation of $R[1 \dd \per(R)]$ equals the Lyndon root.
Over all $\tau$-runs, the total time is~$\cO(n)$ due to \cref{fact:runs_overlap}.
\end{proof}

We next prove the main result of this section.

\begin{theorem}\label{thm:long-ipm}
For any $\ell \in [20\dd \cO(n / \log^2 n)]$, there is a data structure for \textsc{Extended \IPM (Decision)} that can be built using $\cO(n\log_{n/\ell} n)$ + $\cO((n/\ell) \cdot \log^4 n \log \log n)$ time and $\cO((n/\ell) \cdot \log n (\log \log n)^3)$ extra space given random access to $S$ and answers queries in $\cO(\log^3 \log n)$ time, provided that $|P| > \ell$.
\end{theorem}
\begin{proof}
Let $\tau = \lfloor \ell/5 \rfloor$.
We use \cref{thm:part,lemma:runs_compute} with parameter $\tau$ to compute a partitioning set $\P$ of size $\cO(n/\tau)$ and all $\tau$-runs in $S$, grouped by Lyndon root, each one together with the first occurrence of its Lyndon root.
We create a static predecessor structure~$\mathcal{R}$ using \cref{lem:pred}, where we insert the starting position of each run $R$ with the following satellite information: $R$'s ending position, the first occurrence of $R$'s Lyndon root in $R$, and an identifier of its group.
We additionally create a data structure $\mathcal{Q}$, where, for each group of $\tau$-runs with a common root $L$, indexed by their identifiers,
we construct, using \cref{lem:pred}, a predecessor data structure for a set $Q_L := \{(y,s,e) : S[s \dd e] \text{ is the longest }\tau\text{-run with a suffix }L \circ L[1 \dd y]\}$, 
with the first components being the keys and the remaining components being stored as satellite information.
The sets $Q_L$ can be straightforwardly constructed in $\cO(n\log n/\tau)$ time. 

Now, let $\mathcal{L}$ be a set that contains the ending position of each $\tau$-run as well as the starting (resp.~ending) positions of the first (resp.~last) two occurrences of the Lyndon root in this $\tau$-run; $\mathcal{L}$ can be straightforwardly constructed in $\cO(n/\tau)$ time given the information returned by the application of \cref{lemma:runs_compute}.
We then construct a set $\A:= \P \cup \mathcal{L}$ and preprocess the string $S$ and the set~$\A$ according to \cref{lem:sst3d}.

Our query comprises of two steps.

\noindent \textbf{Step 1:}
First, we deal with the case when both $P$ and $T$ have period at most $\tau/4$.
Since $P$ and $T$ are of length at least $5\tau$, due to \cref{fact:runs_overlap}, each of them can be only contained in the $\tau$-run whose starting position is closest to it in the left.
We can thus check whether they both have period at most $\tau/4$ in $\cO(\log\log n)$ time by performing two predecessor queries on $\mathcal{R}$.
If this turns out to be the case, we then check whether the two corresponding $\tau$-runs belong to the same group.
If they do not, then $P$ does not occur in $T$ due to \cref{fact:runs_overlap}.
Otherwise, let the common Lyndon root of the two runs be $L$. 
We can compute in constant time non-negative integers $x_P, x_T, y_P, y_T < |L|$ and $e_P, e_T$ such that $P = L(|L|-x_P \dd] \circ L^{e_P} \circ L[ \dd y_P]$ and $T = L(|L|-x_T \dd] \circ L^{e_T} \circ L[ \dd y_T]$.
Note that $P$ occurs in $T$ if and only if at least one of the following conditions is met: (1) $e_P = e_T$, $x_P \leq x_T$, and $y_P \leq y_T$; or (2) $e_P = e_T - 1$ and $x_P \leq x_T$; or (3) $e_P = e_T - 1$ and $y_P \leq y_T$; or (4) $e_P \leq e_T - 2$. In each of the four cases, we can compute an occurrence of $P$ in $T$ in constant time.

\noindent \textbf{Step 2:}
In the second step of the query, we consider the case when $\per(T) > \tau/4$ and distinguish between two cases depending on whether $\per(S[p \dd p+3\tau]) \leq \tau/4$. In each case, it suffices to perform at most two anchored internal pattern matching queries.

\noindent \textbf{Case I:} $\per(S[p \dd p+3\tau]) > \tau/4$.
Due to Property~\ref{cond:c}, $[p \dd p+3\tau] \cap \P \neq \emptyset$.
Let $x = \min ([p \dd p+3\tau] \cap \P)$.
Additionally, due to Property~\ref{cond:b}, for any occurrence of $P$ in $S$ at position $j$, we have $[p \dd p+3\tau] \cap \P = (p-j) +\left([j \dd j+3\tau] \cap \P\right)$,
and hence $j+(x-p) \in \P$. 
Thus, an anchored \IPM query returns the desired answer in $\cO(\log^3 \log n)$ time.

\noindent \textbf{Case II:} $\per(S[p \dd p+3\tau]) \leq \tau/4$.
We distinguish between two subcases depending on whether $\per(P) > \tau/4$; we can check this in $\cO(\log\log n)$ time with the aid of data structure~$\mathcal{R}$ by comparing $p'$ with the ending position of the $\tau$-run that contains $S[p \dd p+3\tau]$ and checking if $a = P[|P|-\per(S[p \dd p+3\tau])]$ if $a \neq \varepsilon$.

\textbf{Subcase (a)}: $\per(P) > \tau/4$.
In this case, for any occurrence of $P$ in $T$, the ending position of the $\tau$-run that is a prefix of $P$ must be aligned with the ending position of a $\tau$-run in $T$, which belongs to $\mathcal{L} \subseteq \mathcal{A}$.

Recall that $P = S[p \dd p'] a$. 
If the period of $S[p \dd p']$ is greater than~$\tau/4$, we retrieve the ending position of the $\tau$-run containing $S[p \dd p+3\tau]$, which is in~$\mathcal{L} \subseteq \mathcal{A}$ as well and issue an anchored internal pattern matching query. Assume now that the period of $S[p \dd p'+1]$ is at most $\tau/4$ and $\varepsilon \neq a \neq P[|P|-\per(S[p \dd p'])]$, in which case $p'$ might not be in $\mathcal{A}$.
In this case, we retrieve a fragment $S[q \dd q']$ equal to $S[p \dd p']$, such that $q'$ is an ending position of a $\tau$-run in $\cO(\log \log n)$ time using the data structure $\mathcal{Q}$, if such a fragment exists, and use $q \in \mathcal{L} \subseteq \mathcal{A}$ as the anchor to our internal anchor query, effectively searching for $S[q \dd q']a=P$.
Observe that if such a fragment $S[q \dd q']$ does not exist, $P$ cannot have any occurrence in $T$.

\textbf{Subcase (b)}: $\per(P) \leq \tau/4$.
We consider an occurrence of $P$ in the $\tau$-run that contains~$P$ that starts in its first $\per(P)$ positions and one that ends in its last $\per(P)$ positions.
Let these two occurrences be at positions $p_1$ and $p_2$, respectively.
Each of these occurrences contains at least one element of $\mathcal{L}$; let those elements be denoted $q_1$ for the occurrence at $p_1$ and $q_2$ for the occurrence at $p_2$.

Note that these elements can be straightforwardly computed given the endpoints of the $\tau$-run, the endpoints of $P$, and the first occurrence of the Lyndon root in the $\tau$-run, which we already have in hand.
We then issue anchored internal pattern matching queries for $(p_1,q_1,p_1+|P|-1,t,t')$ and $(p_2,q_2,p_2+|P|-1,t,t')$ as both $q_1$ and~$q_2$ are in $\mathcal{L}$.
These queries are answered in $\cO(\log^3 \log n)$ time.
As we show next, if $P$ has an occurrence in $T$, this occurrence will be returned by those queries.

Consider an occurrence of $P$ in $S[t \dd t']$ and denote the $\tau$-run that contains this occurrence by $R$. 
Since $\per(T) > \tau/4$, $R$ does not contain $S[t \dd t']$.
Without loss of generality, let us assume that $R$ does not extend to the left of $S[t \dd t']$, the remaining case is symmetric.
Let the first occurrence of the Lyndon root $L$ of the $\tau$-run in $P$ be at position $i=q_1-p_1+1$ of~$P$, noting that $i\leq \per(P)$. 
Then, in the leftmost occurrence of $P$ in $R$, position $i$ must be aligned with either the first or the second position where $L$ occurs in $R$.
By the construction of the set $\mathcal{L}$, it follows that both of these positions are in $\mathcal{L}$, and hence the anchored internal pattern matching query will return an occurrence.
\end{proof}

\ipm*
\begin{proof}
If the length of $P$ is at most $\max\{\tau,20\}$, we compute its occurrences in~$T$, whose length is $\cO(\tau)$, in $\cO(\tau)$ time using \cref{fact:pm-ro}. In what follows, we assume that $|P|>\max\{\tau,20\}$.

We build the \textsc{Extended \IPM (Decision)} data structure  of \cref{thm:long-ipm} for $S$ with $\ell=\max\{\tau,20\}$.
This allows us to efficiently answer the decision version of the desired \IPM queries, also returning a witness, in $\cO(\log^3\log n)$ time.
If the query does not return an occurrence of $P$ in $T$, we are done.
Otherwise, we have to compute all occurrences of $P$ in~$T$ represented as an arithmetic progression (cf~\cref{cor:progression}). Let the witness returned by the data structure be $S[x \dd x']$. Consider the rightmost occurrence of $P$ in $S[t \dd x')$, or, if it does not exist, the leftmost occurrence in ~$S(x \dd t']$. Such an occurrence can be found by binary search. If no such occurrence exists, we are again done, as $P$ has a single occurrence in $T$. Otherwise, the occurrences of $P$ in $T$ form an arithmetic progression with difference equal to the difference $d$ of $x$ and the starting position of the found occurrence due to \cref{cor:progression}.
We compute the extreme values of this arithmetic progression using binary search as well:
we compute the minimum and the maximum $j \in \mathbb{Z}$ such that $S[x+j\cdot d \dd x'+j\cdot d] = P$ and $t \le x+j\cdot d \le x'+j\cdot d \le t'$ using $\cO(\log n)$ \IPM queries in total; the complexity follows. 
\end{proof}

\subsection{Lower Bound for an \IPM data structure}
We now show that the product of the query time and the space achieved in \cref{cor:ipm} is optimal up to polylogarithmic factors, via a reduction from the following problem.

\defDSproblem{Longest Common Extension (\LCE)}{A string $S$ of length $n$.}{Given $i,j\in[1\dd n]$, return the largest $\ell$ such that $S[i\dd i+\ell] = S[j\dd j+\ell]$.}

Bille et al.~\cite[Lemma 4]{bille2014time} showed that any data structure for \LCE for $n$-length strings that uses $s$ bits of extra space on top of the input has query time $\Omega(n/s)$.

\begin{lemma}\label{lemma:ipm-lb}
  In the non-uniform cell-probe model, any \IPM data structure that uses $s$ bits of space on top of the input for a string of length $n$, has query time~$\Omega(n/(s\log n))$.
\end{lemma}
\begin{proof}
  We prove \cref{lemma:ipm-lb} by reducing \LCE queries in a string $S$ of length $n$ to \IPM queries in $S$. Consider an 
 \IPM data structure with space $s$ and query time $q$ and observe that \IPM queries can be used to check substring equality since $S[i\dd i'] = S[j\dd j']$ if and only if $S[i\dd i']$ occurs inside the interval $[j\dd j']$ and $j'-j = i'-i$.
  Using binary search, we can thus answer any \LCE query via $\cO(\log n)$ \IPM queries.
 Hence, we have $q\log n = \Omega(n/s)$, which concludes the proof the lemma.
\end{proof}

\cref{lemma:ipm-lb} implies a similar lower bound for the word RAM model, which is weaker than the non-uniform cell-probe model.

\section{Other Internal Queries and Approximate Pattern Matching}\label{sec:apps}
In the \pillar model of computation~\cite{9317938} the runtimes of algorithms are analysed with respect to the number of calls made to standard word-RAM operations and a few primitive string operations. It has been used to design algorithms for internal queries~\cite{DBLP:journals/corr/KociumakaRRW13,doi:10.1137/1.9781611973730.36,koc18}, approximate pattern matching under Hamming distance~\cite{9317938} and edit distance~\cite{9996673}, circular approximate pattern matching under Hamming distance~\cite{DBLP:conf/esa/Charalampopoulos22} and edit distance \cite{charalampopoulos2024approximate}, and (approximate) wildcard pattern matching under Hamming distance~\cite{DBLP:journals/corr/abs-2402-07732}.
Space-efficient implementations of the \pillar model immediately result in space-efficient implementations of the above algorithms. 

In the \pillar model, one is given a family of strings $\mathcal{X}$ for preprocessing.
The elementary objects are fragments $X[i\dd j]$ of strings $X \in \mathcal{X}$. Each fragment $S$ is represented via a handle, which is how~$S$ is passed as input to \pillar operations.
Initially, the model provides a handle to each $X \in \mathcal{X}$. 
Handles to other fragments can be obtained through an $\mathsf{Extract}$ operation:
\begin{itemize}
\item $\mathsf{Extract}(S,\ell,r)$: Given a fragment $S$ and positions $1 \le \ell \le r \le |S|$, extract
$S[\ell \dd r ]$. 
\end{itemize}
Furthermore, given elementary objects $S, S_1, S_2$ the following primitive operations are supported in the \pillar model:
\begin{itemize}
\item $\mathsf{Access}(S,i)$: Assuming $i \in [1\dd |S|]$, retrieve $S[i]$.
\item $\mathsf{Length}(S)$: Retrieve the length $|S|$ of $S$.
\item Longest common prefix $\LCE(S_1,S_2)$: Compute the length of the longest common prefix of $S_1$ and $S_2$.
\item $\LCE^R(S_1,S_2)$: Compute the length of the longest common suffix of $S_1$ and $S_2$.
\item Internal pattern matching $\IPM(S_1,S_2)$: Assuming that $|S_2| < 2|S_1|$, compute the set of the starting positions
of occurrences of $S_1$ in $S_2$ represented as one arithmetic progression.
\end{itemize}

All \pillar operations other than $\LCE$, $\LCE^R$, and \IPM admit trivial constant-time and constant-space implementations in the read-only setting. For any $\tau = \cO(n/\log^2 n)$,
Kosolobov and Sivukhin \cite{DBLP:journals/corr/abs-2105-03782} showed that after $\cO(n \log_{n/\tau}n)$-time, $\cO(n/\tau)$-space preprocessing, $\LCE$ and $\LCE^R$ queries can be supported in $\cO(\tau)$ time.
For \IPM queries, we use \cref{cor:ipm}.

In~\cite{DBLP:journals/corr/KociumakaRRW13,doi:10.1137/1.9781611973730.36,koc18} it is (implicitly) shown that the following internal queries can be efficiently implemented in the \pillar model.
\begin{itemize}
\item A \emph{cyclic equivalence query} takes as input two equal-length fragments $U = S[i\dd i+\ell]$ and $V = S[j \dd j+\ell]$, and returns all rotations of $U$ that are equal to $V$.
Any cyclic equivalence query reduces to $\cO(1)$ \LCE queries and $\cO(1)$ $\IPM(P,T)$ queries with $|T|/|P|=\cO(1)$.
\item A \emph{period query} takes as input a fragment $U = S[i \dd j]$, and returns all periods of $U$. Such a period query reduces to $\cO(\log |U|)$ \LCE queries and $\cO(\log |U|)$ $\IPM(P,T)$ queries with $|T|/|P|=\cO(1)$.
\item A \emph{2-period} query takes as input a fragment $U = S[i \dd j]$, checks if $U$ is periodic and, if so, it also returns $U$'s period.
Such a query reduces to $\cO(1)$ \LCE queries and $\cO(1)$ $\IPM(P,T)$ queries with $|T|/|P|=\cO(1)$.
\end{itemize}

\begin{corollary}
Suppose that we have read-only random access to a string $S$ of length $n$ over an integer alphabet.
For any integer $\tau = \cO(n/\log^2 n)$, there is a data structure that can be built using $\cO(n \log_{n/\tau} n + (n/\tau) \cdot \log^4 n \log \log n)$ time and $\cO((n/\tau) \cdot \log n (\log \log n)^3)$ extra space and can answer cyclic equivalence and 2-period queries on $S$ in $\cO(\tau+\log n \log^3 \log n)$ time, and period queries on $S$ in $\cO(\tau \log n +\log^2 n \log^3 \log n)$ time.  
\end{corollary}

By plugging this implementation of the \pillar model into~\cite{9317938,9996673,DBLP:conf/esa/Charalampopoulos22,DBLP:journals/corr/abs-2402-07732,charalampopoulos2024approximate}, we obtain the following:

\begin{corollary}
    Suppose that we have read-only random access to a text $T$ of~length $n$, a pattern $P$ of~length $m$ over an integer alphabet.
 Given an integer threshold $k$, for any integer $\tau = \cO(m/\log^2m)$,
    we can compute:
    \begin{itemize}
    \item the approximate occurrences of $P$ in $T$ under the Hamming distance in $\cOtilde(n+k^2\tau \cdot n/m)$ time using $\cOtilde(m/\tau+k^2)$ extra space;
    \item the approximate occurrences of $P$ in $T$ under the edit distance in $\cOtilde(n+k^{3.5}\tau \cdot n/m)$ time using $\cOtilde(m/\tau+k^{3.5})$ extra space;
    \item the approximate occurrences of all rotations of $P$ in $T$ under the Hamming distance in $\cOtilde(n+k^3\tau \cdot n/m)$ time using $\cOtilde(m/\tau+k^3)$ extra space;
    \item the approximate occurrences of all rotations of $P$ in $T$ under the edit distance in $\cOtilde(n+k^5\tau \cdot n/m)$ time using $\cOtilde(m/\tau+k^5)$ extra space;
    \item in the case where $P$ has $D$ wildcard letters arranged in $G$ maximal intervals, the approximate occurrences of $P$ in $T$ under the Hamming distance in $\cOtilde(n+ (D+k)(G+k)\tau \cdot n/m)$ time using $\cOtilde(m/\tau+(D+k)(G+k))$ extra space.
    \end{itemize}
    \end{corollary}

To the best of our knowledge, the only work that has considered approximate pattern matching in the read-only model is due to Bathie et al.~\cite{DBLP:conf/isaac/BathieKS23}. They presented online algorithms both for the Hamming distance and the edit distance; for the Hamming distance their algorithm uses $\cO(k \log m)$ extra space and $\cO(k \log m)$ time per letter of the text, and for the edit distance $\tO(k^4)$ bits of space and $\tO(k^4)$ amortised time per letter. 

%%%%%%%%
%%%%%%%%
\section{\LCS and \CPM in the Streaming Setting}
In the streaming setting, we receive a stream composed of the concatenation of the input strings, e.g., the pattern and the text in the case of \CPM. We account for all the space used, including the space needed to store any information about the input strings. We exploit the well-known connection between streaming algorithms and communication complexity to prove linear-space lower bounds for streaming algorithms for \LCS and \CPM.

\subsection{Lower Bounds for Streaming Algorithms}

Our streaming lower bounds are based on a reduction from the following problem:

\defCCproblem{Augmented Index}{a binary string $S$ of length $n$.}{an index $i \in [1\dd n]$ and the string $S[\dd i-1]$.}{Bob is to return the value of $S[i]$.}

In the one-way communication complexity model, Alice performs an arbitrary computation on her input to create a message $\mathcal{M}$ and sends it to Bob who must compute the output using this message and his input. The communication complexity of a protocol is the size of $\mathcal{M}$ in bits.
The protocol is randomised when either Alice or Bob use randomised computation. 
\begin{theorem}[{\cite[Theorem 2.3]{DBLP:journals/siamcomp/ChakrabartiCKM13}}]\label{thm:index-lb}
  The randomised one-way communication complexity of \AI is $\Omega(n)$ bits.
\end{theorem}

\begin{theorem}\label{thm:lb-lcs}
  In the streaming setting, any algorithm for $\LCS$ for strings of length at most~$n$ uses $\Omega(n)$ bits of space.
\end{theorem}
\begin{proof}
  We show the bound by a reduction from the \AI problem.
  Consider an input $S, (i, S[\dd i-1])$ to the \AI problem, where $|S| = n$.
  We observe that for $A = 0^n\$S$ and $B = 0^n\$S[\dd i-1]1$, where $\$\notin\{0, 1\}$,
    we have $\LCS(A, B) = n+ i+1$ if and only if $S[i] = 1$. 
 Now, if we have a streaming algorithm for \LCS that uses $b$ bits of space, we can develop a one-way protocol for the \AI problem with message size $b$ bits as follows. Alice runs the algorithm on $A$. When she reaches the end of $A$, she sends the memory state of the algorithm and $n$ (in binary) to Bob. Bob continues running the algorithm on $B$, which he can construct knowing $n$ and $S[\dd i-1]$, and returns 1 if and only if $\LCS(A, B) = n+i+1$. \cref{thm:index-lb} implies that $b + \log n = \Omega(n)$, and hence $b = \Omega(n)$.
\end{proof}

\begin{theorem}\label{th:streaming-lb-cpm}
  In the streaming setting, any algorithm for \CPM uses $\Omega(m)$ bits of space, where $m$ is the size of the pattern. 
\end{theorem}
\begin{proof}
  We show the bound by a reduction from the \AI problem.
  Consider an input $S, (i, S[\dd i-1])$ to the \AI problem, where $|S| = m$.
    Let $A = S\$$ and $B = S\$S[\dd i-1]1$, where $\$\notin\{0, 1\}$.
    $B$ ends with an occurrence of a rotation of $A$ if and only if $S[i] = 1$.
  Now, if we have a streaming algorithm for \CPM that uses $b$ bits of space, we can develop a one-way protocol for the \AI problem with message size $b$ bits as follows.
  Alice runs the algorithm on the pattern $A = S\$$ and the first $|S|+1$ letters of the string $B$.
  She then sends the memory state of the algorithm to Bob. Bob continues running the algorithm on the remainder of $B$, i.e., on $S[1 \dd i-1] 1$, and returns 1 if and only if the algorithm reports an occurrence of a rotation of $A$ ending at position $n+i+1$.
  By \cref{thm:index-lb}, we have $b = \Omega(m)$.
\end{proof}

%%%%%%%%
%%%%%%%%
\section{\LCS and \CPM in the Asymmetric Streaming Setting}

In this section, we use \cref{thm:long-ipm} to show that for any $\tau\in[\tilde{\Omega}(\sqrt{m}) \dd \cO(m/\log^2m)]$, there are asymmetric streaming algorithms for \LCS and \CPM that use $\cO(\tau)$ space and $\cOtilde(m/\tau)$ time per letter. We start by giving an algorithm for a generalization of the \LCS problem that can be used to solve both \LCS and \CPM.
For two strings $S,T$, a fragment $T[t\dd t']$ is a \emph{$T$-maximal common substring} of $S$ and $T$ if it is a occurs in $S$ and neither $T[t-1\dd t']$ (assuming $t > 1$) nor $T[t\dd t'+1]$ (assuming $t' < n$) occurs in $S$.

\begin{theorem}\label{thm:lcs-aux}
Assume to be given read-only random access to a string $S$ of length $m$ and streaming access to a string $T$ of length $n$ over an integer alphabet, where $n \ge m$. 
  For all $\tau\in[\sqrt{m} \log m (\log \log m)^3 \dd \cO(m/\log^2m)]$, there is an algorithm that reports all $T$-maximal common substrings of $S$ and $T$ using $\cO(\tau)$ space and $\cO(nm/\tau \cdot \log\log \sigma)$ time. 
\end{theorem}
\begin{proof}
  We cover $T$ with windows of length $2\tau$ (except maybe for the last) that overlap by~$\tau$ letters: there are $\cO(n/\tau)$ such windows.
  After reading such a window $W$, we apply the procedure encapsulated in the following claim with $A = W$ and $B = S$:
  \begin{claim}\label{claim:lcs-lcsuff} 
    Let $A,B$ be strings of respective lengths $a$ and $b$, where $a < b < a^{\cO(1)}$, over an integer alphabet of size $\sigma$.
    Given read-only random access to $A$ and $B$, we can compute all $B$-maximal common substrings of $A$ and $B$,
    and the length $\LCSu(A,B)$ of the longest suffix of $A$ that occurs in $B$ in $\cO(b \log\log \sigma)$ time using $\cO(a)$ extra space.
  \end{claim}
  \begin{claimproof}
    We start by building the suffix tree for $A$ and preprocessing it for constant-time weighted ancestor queries: this takes $\cO(a)$ time (see \cref{thm:sst} and \cref{lem:WAQ-st}).  Additionally, we preprocess the labels  of edges outgoing from each node according to \cref{lem:pred}.
    Then, the algorithm traverses the tree maintaining the following invariant: at every moment, it is at a node (maybe implicit) corresponding to a substring $B[i \dd j]$ of $B$. It starts at the root of the tree with $i = 1$ and $j = 0$. In each iteration, the algorithm tries to go down the tree from the current node using $B[j+1]$; this takes $\cO(\log \log \sigma)$ time.
    If it succeeds, it increments $j$ and continues. Otherwise, it considers two cases. If it is at the root, it increments both $i$ and $j$. Otherwise, it jumps to the node corresponding to $B[i+1 \dd j]$ via a weighted ancestor query in $\cO(1)$ time and increments $i$. The nodes reached by an edge traversal and abandoned with the use of a weighted ancestor query in the next iteration are in one-to-one correspondence with the $B$-maximal common substrings of $A$ and $B$. The \LCSu of $A$ and $B$ is the depth of the deepest  visited node that corresponds to a suffix of $A$.
    As at least one of the indices $i,j$ gets incremented at every step of the traversal, the total runtime is $\cO(b\log\log\sigma)$. 
  \end{claimproof}
  The above sliding-window procedure takes $\cO(m\log\log \sigma)$ time per window and uses $\cO(\tau)$ space,
  which adds up to $\cO(n m/\tau\cdot \log\log \sigma)$ time in total, and finds all $T$-maximal common substrings of~$S$ and~$T$ that have length at most $\tau$.

  We run another procedure in parallel in order to compute $T$-maximal common substrings of length at least $\tau$. During preprocessing, we build the \textsc{Extended \IPM (Decision)} data structure (\cref{thm:long-ipm}) for the string $S$ with $\ell = \tau - 2$ in $\cO(m \log_{m/\tau} m) = \cO(nm/\tau)$ time using $\cO((m/\tau) \cdot \log m (\log \log m)^3) = \cO(\tau)$ space. 

Assume that while reading a window $W = T[\ell \dd r]$, the sliding-window procedure found an \LCSu $T[i \dd r]$ of length at least $\tau$. We start a search for a common substring starting in~$W$. Let $j \ge r$ be the current letter of $T$, and $T[i \dd j]$, $\ell \le i \le r$, be the longest suffix of $T[\ell \dd j]$ that occurs in~$S$.
We assume that we know a position where $T[i \dd j]$ occurs in~$S$, which is the case for $j = r$. When $T[j+1]$ arrives, we update $i$ using the following observation:
  
\begin{observation}  
If $T[i \dd j]$ is the longest suffix of $T[1 \dd j]$ that occurs in $S$, and $T[i' \dd j+1]$ is the longest suffix of $T[1 \dd j+1]$ that occurs in $S$, then $i \le i'$. 
\end{observation}
  
By using binary search and \IPM queries, we can find the smallest $i \le i'$ such that $T[i' \dd j+1]$ occurs in $S$ and a witness occurrence, if the corresponding string has length at least $\tau$: namely, if $S[x \dd x']$ is a witness occurrence of $T[i \dd j]$ in $S$, we search for occurrences of $P = S[x + (i'-i) \dd x'] T[j+1]$ in $S$.
If $j-i' < \tau$, we stop the search, and otherwise we set $i' = i$ and continue.
It is evident that all $T$-maximal common substrings of $S$ and $T$ that are of length greater than $\tau$ can be extracted during the execution of the above procedure:
a maintained suffix of length greater than $\tau$ is such a fragment if the last update to it was an increment of its right endpoint, while the next update is an increment of its left endpoint.
For every letter, we run at most one binary search which uses $\cO(\log m)$ \IPM queries and hence takes $\cO(\log m (\log \log m)^3)$ time. 
As $\tau = \cO(m/\log^2 m)$, the $m/\tau$ term dominates the per-letter running time.
  The correctness of the described procedure follows from the fact that any substring of $T$ of length greater than~$\tau$ is either fully contained in the first window or crosses the boundary of some window.
\end{proof}

\begin{corollary}\label{cor:lcs}
Assume to be given random access to an $m$-length string $S$ and streaming access to a $n$-length string $T$, where $n \ge m$. 
  For all $\tau\in[\sqrt{m} \log m (\log \log m)^3 \dd \cO(m/\log^2m)]$, there is an algorithm that computes $\LCS(S,T)$ using $\cO(nm/\tau\cdot \log\log \sigma)$ time and $\cO(\tau)$ space.
\end{corollary}
\begin{proof}
Note that the longest common substring of $S$ and $T$ is a $T$-maximal substring of $S$ and $T$.
We use the algorithm of \cref{thm:lcs-aux} with the same value of $\tau$ to iterate over all $T$-maximal common substrings $T[t\dd t']$ of $S$ and $T$, and store the pair of indices $t,t'$ that maximizes $t'-t$.
\end{proof}

\begin{corollary}\label{cor:cpm}
Assume to be given random access to an $m$-length pattern $P$ and streaming access to an $n$-length text $T$, where $n \ge m$. 
For all $\tau\in[\sqrt{m}\log m (\log \log m)^3 \dd \cO(m/\log^2m)]$, there is an algorithm that solves the \CPM problem for $P,T$ using $\cO(m/\tau\cdot \log\log \sigma)$ time per letter of~$T$ and $\cO(\tau)$ space.
\end{corollary}
\begin{proof}
We use the algorithm of \cref{thm:lcs-aux} with threshold $\tau$ on the string $P \cdot P$, to which we have random access, and a streaming string $T$. The occurrence of any rotation of $P$ in~$T$ implies a common substring of $P \cdot P$ and $T$ of length $m \ge 2 \tau$. The algorithm of \cref{thm:lcs-aux} allows us to find such occurrences in $\cO(m/\tau\cdot \log\log \sigma)$ amortized time per letter of $T$ using $\cO(\tau)$ space. By noticing that none of the $m$-length substrings are fully contained in $T(|T|-\tau \dd |T|]$, we can deamortise the algorithm using the standard time-slicing technique, cf~\cite{DBLP:journals/iandc/CliffordEPP11}.
\end{proof}

%%%%%%%%
%%%%%%%%
\section{\CPM in the Read-only Setting}
In this section, we present a deterministic read-only online algorithm for the \CPM problem.

\begin{theorem}\label{thm:readonly}
There is a deterministic read-only online algorithm that solves the \CPM problem on a pattern $P$ of length $m$ and a text $T$ of length $n$ using $O(1)$ space and $O(1)$ time per letter of the text.
\end{theorem}
\begin{proof}
In this proof, we assume that $n \le 2m-1$. If this is not the case, we can cover~$T$ with $2m$-length windows overlapping by $m-1$ letters, and process the text window by window; the last window might be shorter. Every occurrence of a rotation of $P$ belongs to exactly one of the windows and hence will be reported exactly once.

We partition $P$ into four fragments $P_1, P_2, P_3, P_4$, each of length either $\lfloor m/4 \rfloor$ or~$\lceil m/4 \rceil$.\footnote{The sole reason for partitioning $P$ into four fragments instead of two is to guarantee that there is an occurrence of some $P_i$ close the the starting position of each rotation of $P$. This allows us to obtain a worst-case rather than an amortised time bound for processing each letter of the text.}
By applying \cref{lem:run_per}, we compute the periods of each of $P$ and $P_i$ for $i\in [1\dd 4]$, if it is are periodic.
We also compute, for each $i\in [1\dd 4]$, the occurrences of $P_i$ in $P^2$ using \cref{fact:pm-ro}, and store them in $\cO(1)$ space due to \cref{cor:progression}.
Overall, the preprocessing step takes $\cO(m)$ time and uses constant space.

We compute all occurrences of all $P_i$ in $T$ in an online manner using \cref{fact:pm-ro}.
Due to \cref{cor:progression}, we can represent all computed occurrences of each $P_i$ using a constant number of arithmetic progressions with difference $\per(P_i)$ in $\cO(1)$ space.
\begin{observation}\label{obs:duality}
Assume that $T(j-m\dd j] = P[\Delta+1\dd m] \circ P[\dd \Delta]$.
There is an occurrence of $P_i$ at a position~$\ell$ of $T$ such that $j-m < \ell \le j-|P_i|+1$ if and only if there is an occurrence of $P_i$ at position $p=\Delta+\ell-j+m$ of $P^2$.
\end{observation}
Now, note that for every rotation $P'$ of $P$, some $P_i$ occurs at one of the first $\phi := 2\lceil m/4 \rceil$ positions of $P'$.
We will use such occurrences as anchors to compute the occurrences of rotations of $P$ in $T$.
Fix $i$ such that there is an occurrence of $P_i$ in the first $\phi$ positions of $T(j-m\dd j]$.
We consider two cases depending on whether the period of $P_i$ is large or small. 

\textbf{Case I:} $\per(P_i) > |P_i|/4$. By \cref{cor:progression}, there are $O(1)$ occurrences of $P_i$ in each of $T$ and~$P^2$.
Suppose that $P_i$ occurs at position~$\ell$ of~$T$.
If $T(j-m\dd j] = P[\Delta+1\dd m] \circ P[\dd \Delta]$ for some $\Delta$, then, by~\cref{obs:duality}, $P_i$ occurs at position $p=\Delta+\ell-j+m$ of $P^2$ and we must have that the length of the longest common suffix of $T[1 \dd \ell)$ and $P^2[1 \dd p)$ is at least $\ell-(j-m)$ and the length of the longest common prefix of $T[\ell+|P_i| \dd]$ and $P^2[p+|P_i|\dd]$ is at least $j-\ell-|P_i|$.
As we only need to consider occurrences of $P_i$ in the first $\phi$ positions of rotations of $P$, we can work under the assumption that $\ell-(j-m) \leq \phi$.
Hence, it suffices to compute, for every occurrence of $P_i$ at a position $p$ in $P^2$ and every occurrence of $P_i$ at a position $\ell$ in $T$, values
\begin{itemize}
\item $x := \max\{\phi, \LCE^R(T[1 \dd \ell),P^2[1 \dd p))\}$, the maximum of $\phi$ and the length of the longest common suffix of $T[1 \dd \ell)$ and $P^2[1 \dd p)$;
\item $y := \LCE^R(T[1 \dd \ell),P^2[1 \dd p))$, the length of the longest common prefix of $T[\ell+|P_i| \dd]$ and $P^2[p+|P_i|\dd]$.
\end{itemize} 
The length $y$ is computed naively as new letters arrive, while, in order to compute $x$, we perform a constant number of letter comparisons for each letter that arrives.
Since $\ell-(j-m)=\cO(j-\ell-|P_i|)$, we will have completed the extension to the left when the $j$-th letter of the text arrives.
As there is a constant number of pairs $(p,\ell)$ to be considered, we perform a total number of $O(1)$ letter comparisons per letter of the text.
 
\textbf{Case II:} $\per(P_i) \le |P_i|/4$. For brevity, denote $\rho = \per(P_i)$.
Below, when we talk about arithmetic progressions of occurrences of $P_i$, we mean maximal arithmetic progressions of starting positions of occurrences of $P_i$ with difference $\rho$.
Consider the first element~$\ell$ and the last element $r$ of the rightmost computed arithmetic progression of occurrences of $P_i$ in $T(j-m\dd j]$.
We next distinguish between two cases depending on whether $\per(T(j-m \dd j])=\rho$.
This information can be easily maintained in $\cO(1)$ time per letter using $\cO(1)$ space as follows.
In particular, for each arithmetic progression of occurrences of~$P_i$ in $T$, we perform at most $\rho-1$ letter comparisons to extend the periodicity to the left; we can do this lazily upon computing the first element of each progression, by performing at most one letter comparison for each of the next $\rho-1$ letter arrivals.
Further, as at most one arithmetic progression corresponds to occurrences of~$P_i$ in $T$ that contain a position in~$(j-\rho \dd j]$, the extensions to the right take $\cO(1)$ time per letter as well.

\textbf{Subcase (a): $\per(T(j-m \dd j]) \neq \rho$.}
Suppose that $T(j-m\dd j] = P[\Delta+1\dd m] \circ P[\dd \Delta]$ for some $\Delta$.
Then, due to \cref{obs:duality}, one of the two following holds:
\begin{enumerate}
\item $\ell$ and $p_\ell =\Delta+\ell-j+m$ are the first elements in arithmetic progressions of occurrences of $P_i$ in $T(j-m \dd j]$ and $P^2$, respectively;
\item $r$ and $p_r =\Delta+r-j+m$ are the last elements in arithmetic progressions of occurrences of $P_i$ in $T(j-m \dd j]$ and $P^2$, respectively.
\end{enumerate}
We handle this case by considering a subset of pairs of occurrences of $P_i$ and treating them similarly to Case I. Namely, we consider (a) pairs that are first in their respective arithmetic progressions in $P^2$ and $T$ and (b) pairs that are last in their respective arithmetic progressions in $P^2$ and $T(j-m \dd j]$.
By \cref{cor:progression}, there are only a constant number of such elements in~$P^2$ and a constant number of such elements in the text at any time (a previously last element in the text may stop being last when a new occurrence of $P_i$ is detected).
For pairs of first elements there are no changes required to the algorithm for Case I. We next argue that, for each pair $(r,p_r)$ of last elements, it suffices to perform only $\cO(\rho)$ letter comparisons to check how far the periodicity extends to the left, and that this is all we need to check.
Due to this, we do not restrict our attention to the case when $r \in (j-m \dd j-m+\phi]$, but rather consider all last elements of arithmetic progressions.
Let $\ell'$ be the first element of the arithmetic progression in $T(j-m \dd m]$ that contains $r$.
If $\ell' > \rho+j-m$, we avoid extending to the left since either $\ell' \in (j-m \dd j-m+\phi]$ and the sought occurrence of a rotation of $P$, if it exists, will be computed by the algorithm when it processes pair $(\ell',\Delta+\ell'-j+m)$ or the sought occurrence will be computed when processing a different arithmetic progression of occurrences of $P_i$ or a different $P_j$.
Further note that the extension to the left has been already computed; either $\ell'$ is not the first element in the arithmetic progression of occurrences of $P_i$ in $T$ (we have assumed that it is in $T(j-m \dd j]$), in which case we are trivially done, or $\ell'$ is the first element of an arithmetic progression in $T$ and hence we extended the periodicity via a lazy computation when the occurrence of $P_i$ at position $\ell'$ was detected. As the occurrences of $P_i$ in $T$ are spaced at least $\rho$ positions away, the above procedure takes $\cO(1)$ time per letter of the text.

\textbf{Subcase (b): $\per(T(j-m \dd j]) = \rho$.}
Using $\cO(m)$ time and $\cO(1)$ extra space, we can precompute all $1 \le j \le \rho$ such that $Q_i^\infty[j \dd j+m)$ occurs in $P^2$, where $Q_i = P_i[1 \dd \rho]$; it suffices to extend the periodicity for each of the $\cO(1)$ arithmetic progressions of occurrences of $P_i$ in $P^2$ and to perform standard arithmetic.
In particular, the output consists of a constant number of intervals.
Then, if $\per(T(j-m \dd j]) = \rho$, $T(j-m \dd j]$ equals a rotation of $P$ if and only if $\ell - (j-m) \pmod \rho$ is in one of the computed intervals and this can be checked in constant time.
\end{proof}

%%%%%%%%%%%%%%%%%%%%%%%%%%%%
\bibliographystyle{plainurl}
\bibliography{main.bib}

\end{document}